\RequirePackage{fix-cm}
\documentclass[smallextended]{svjour3}       
\smartqed  
\usepackage{graphicx}
\usepackage{stfloats}
\usepackage{subcaption}
\usepackage[hyphens]{url}  
\usepackage{amssymb}
\usepackage{amsmath}
\usepackage{color}
\usepackage{xcolor}
\begin{document}

\title{ Discouraging Pool Block Withholding Attacks in Bitcoins 
}


\author{Zhihuai Chen\and
    Bo Li\and
    Xiaohan Shan \and
    Xiaoming Sun \and
    Jialin Zhang
}

\authorrunning{Chen \and Li \and Shan \and Sun \and Zhang} 

\institute{Zhihuai Chen \and Jialin Zhang \and Xiaoming Sun \at
    University of Chinese Academy of Sciences, 100049, Beijing, China
    \at CAS Key Lab of Network Data Science and Technology, Institute of Computing Technology, Chinese Academy of Sciences (CAS), 100190 Beijing, China\\
    \email{\{chenzhihuai, zhangjialin, sunxiaoming\}@ict.ac.cn}           
    \and
    Bo Li \at
    Department of Computer Science, University of Oxford, OX1 3QD, Oxford, United Kingdom
    \email{boli@cs.ox.ac.uk}           
    \and
    Xiaohan Shan \at
    Department of Computer Science and Technology, Tsinghua University, 100084, Beijing, China\\
    \email{shanxiaohan@tsinghua.edu.cn}           
}

\date{Received: date / Accepted: date}

\maketitle

\begin{abstract}
    The arisen of Bitcoin has led to much enthusiasm for blockchain research and block mining,
    and the extensive existence of mining pools helps its participants (i.e., miners) gain reward more frequently.
    Recently, the mining pools are proved to be vulnerable for several possible attacks, and pool block withholding attack is one of them:
    one strategic pool manager sends some of her miners to other pools and
    these miners pretend to work on the puzzles but actually do nothing.
    And these miners still get reward since the pool manager can not recognize these malicious miners.

    In this work, we revisit the game-theoretic model for pool block withholding attacks
    and propose a revised approach to reallocate the reward to the miners.
    Fortunately, in the new model, the pool managers have strong incentive to not launch such attacks.
    We show that for any number of mining pools, no-pool-attacks is always a Nash equilibrium.
    Moreover, with only two minority mining pools participating, no-pool-attacks is actually the unique Nash equilibrium.

\end{abstract}

\section{Introduction}
Bitcoin is a decentralized crypto currency originally proposed by Satoshi Naka-moto in 2008 \cite{bitcoin2008}. Since then, it has attracted wide attention in the world both academically and commercially because of its soaring price and the characteristics of decentralization. Bitcoin realizes its decentralization by \textit{blockchain}, which is a global  ledger maintained by distributed system. This ledger can record historical transactions and other information.

One main difference between Bitcoin and most existing currencies is the way of currency issuance and the method of bookkeeping.
Bitcoin is implemented on a P2P network and everyone is able to join or leave the network without any permission.
The participants in generating new blocks are called \textit{miners},
and their tasks are 
to verify the legitimacy of undetermined transactions and pack those legal transactions into a block.
In order to validate the block, every miner needs to work on a cryptographic puzzle,
which needs large amount of computational resources (like electricity and hardware).
To motivate the miners to generate new blocks, the first miner who solves the puzzle correctly, will get rewarded.
Roughly speaking, the probability that a miner can propose a block successfully is proportional to her computational power among all miners.
The above protocol of identifying the block generator is called \textit{proof of work} and the process of solving puzzles is called \textit{mining}.
More detailed of Bitcoin system description can be referred to the
Bitcoin white paper \cite{bitcoin2008} and wiki\footnote{\url{https://wikipedia.org/wiki/Bitcoin}}.

To control the frequency of block generation, the difficulty of puzzles is adjusted dynamically in Bitcoin system,
which is about every ten minutes in expectation.
With the development of Bitcoin, the computational power in the system is extremely large nowadays.
Thus, it is very difficult for a solo miner to propose a block successfully.
This means the miner may spend several years to obtain (huge) rewards which is of course unacceptable.
To pursuit stable reward, a lot of miners join in the \textit{mining pools} which are gatherings of individual computational powers \cite{Rosenfeld2011bwa}.
Miners in the same mining pool work together on the proof of work protocol,
thus a pool can obtain reward much more frequently compared to mining individually.
After obtaining reward, the mining pool will distribute the reward among miners in this pool based on the computational power of each miner.
For a solo miner, attending a mining pool cannot increase the expected reward, but the variance is improved significantly.
Nowadays, there are more than 12 Bitcoin mining pools and more than 90\% of Bitcoin mining is done
by those pools\footnote{Bitcoin hashrate distribution. \url{https://blockchain.info/pools}. Accessed June 14, 2019.}.

However, it is shown in \cite{Eyal2015pbwa} that the permissionless mining pools
have strong incentives to launch the {\em pool block withholding attacks (PBW attacks)} on each other:
one strategic pool manager sends some of her miners to other pools
and these miners pretend to work on the puzzles but actually do nothing.
In Bitcoin or any decentralized system, the pool managers are not able to recognize such malicious miners,
thus these miners can still obtain the reward from mining pool proportional to their computing powers.
Eyal \cite{Eyal2015pbwa} proved that although infiltrating the other pools consumes some of her computational power,
the reward from infiltrating miners may still increase the pool's total utility.
PBW attack is also demonstrated in \cite{luu2015power} with a different reward function.

Eyal \cite{Eyal2015pbwa} modeled the above scenario as {\em pool block withholding games (PBW games)}
and showed that with any number of pools, no-pool-attacks is not a Nash equilibrium.
When there are two pools, the situation faced by the two managers is similar to prisoner's dilemma,
which is called {\em miner's dilemma}: in an equilibrium,
both manages launch the PBW attack and accordingly earn less rewards
compared to when they mine honestly.

Alkalay-Houlihan and Shah \cite{AlkalayHoulihan2019aaai} further studied
the miner's dilemma between two strategic mining pools and
obtained the bound of the social loss due to noncooperation, i.e., price of anarchy.
They showed that the pure Nash equilibrium always exists, and the pure price of anarchy is at most 3 in this game.
They also conjectured the tight bound should be 2 and demonstrated this in some special cases.

It is disappointing that launching PBW attacks to be always profitable for the pools, and would dramatically decrease the social welfare.
Such bad news is derived from the inaction of the malicious Bitcoin pools.
Actually, in both \cite{Eyal2015pbwa} and \cite{AlkalayHoulihan2019aaai},
the authors assume that all the rewards are proportionally distributed to the miners.
Accordingly, a natural question to ask would be

\begin{quote}
    \em Is there any other approach to reallocate the rewards such that PBW attacks can be avoided or discouraged.
\end{quote}

\subsection{Main Contributions}


Briefly, we resolve the above question affirmatively by refining the PBW games
and show that for any number of mining pools, no-pool-attacks is always a Nash equilibrium under a proper \textit{way to award the successful miner}.
Moreover, with only two minority pools participating, no-pool-attacks is the unique Nash equilibrium.
We summarize our contributions as follows.
\begin{itemize}
    \item We refine the PBW game by allowing all pool managers to deduct (different) percentages of their rewards
        before they distribute the reward to the miners, where the deducted reward is used to award the successful miner.
        We call the refined model DPBW game.
        We believe such a new reward distribution rule in DPBW games is more realistic than the proportional way,
        as, intuitively, the former provides enough motivation for the miners to work hard on mining.
        Moreover, it should also be careful that the manager cannot deduct too much from the total reward since the more they
        deduct, the less incentives are left to the miners to join the pools.

    \item We prove that for arbitrary number of pools, no-pool-attacks is always a Nash equilibrium of DPBW games
        with reasonable deductions.
        Thus the price of stability (PoS) of DPBW games is 1.
        This is intriguing for the situations in which there is some authority that can influence the managers a bit
        and help them converge to the optimal solution that is also stable.

    \item Particularly, for two minority pools participating, the unique Nash equilibrium of DPBW games with reasonable deductions is no-pool-attacks.
        This result is more exciting for the special case: even without such authorities,
        every manager is automatically willing to not launch PBW attacks.
\end{itemize}

Thus in the present work, we improve the negative results proved in \cite{Eyal2015pbwa} and \cite{AlkalayHoulihan2019aaai}
by designing a new way to allocate the reward to the miners and award the successful miner
such that all the computing power will be used to mining.

%
%
%

\subsection{Related Works}

Recent years have seen a number of studies on security issues of blockchain, such as the selfish mining attack \cite{Eyal2014selfishmining}, the Eclipse attack \cite{Heilman2015129eclipse} and the distributed denial-of-service attack \cite{Johnson2014ddos}.
A closely relevant attack to our study is the {\em Block Withholding Attack (BW attack)} which is first defined by Rosenfeld in \cite{Rosenfeld2011bwa}.
In this attack, a miner who has found a legal block chooses not to submit it to the mining pool immediately
but rather delays to submit or even directly abandons it.
Many follow-up studies focus on the simulation and countermeasures of BW attacks.

To resist BW attacks, Schrijvers et al. \cite{schrijvers2016incentive} designed a new incentive compatible reward function and
showed proportional mining rewards are not incentive compatible.
It was shown in \cite{Bag2016preventBWA} that
by giving extra reward to a miner who actually finds the winning block on behalf of the pool,
it is possible to discourage BW attacks.
Tosh et al. \cite{Tosh2017security} modeled the BW attacks in a blockchain cloud and demonstrates that
attacker's access to extra computational power could disrupt the honest mining operation.
Very recently,
Wu et al. \cite{WU2019Equilibrium} constructed a generalized model where two participants can choose to either cooperate with each other or employ a BW attack.
They showed that increasing the information asymmetry by utilizing information conceal mechanisms could lower the occurrence of BW attacks.
More studies in this line can be found in \cite{Bag2017,Mousavinejad2018detect,Haghighat2019}.

The Pool Block Withholding Attack (PBW attack) is a variation of BW attacks, the difference is that the PBW attack is launched by a pool manager rather than a miner.
Eyal \cite{Eyal2015pbwa} describes PBW attacks from the perspective of game theory in which the players are pool managers and the strategy of each manager is allocating some powers to launch the block withholding attack. Eyal also proves that with any number of pools, all pools mining honestly and not attacking each other is never a Nash equilibrium.
Subsequently, \cite{AlkalayHoulihan2019aaai} follows Eyal's study on the case of two pools attacking each other. They show that this game always admits a pure Nash equilibrium, and its pure price of anarchy is at most 3.

Finally, there are also many other game theoretic studies about Bitcoin and other crypto currencies.
For example,
\cite{babaioff2012bitcoin} studies how to incentivize the users to propagate each transaction in a tree network;
\cite{carlsten2016instability} shows that with only transaction fees (and negligible block rewards),
the miners have strong incentive to fork a block and generate a greater reward;
\cite{lewenberg2015bitcoin} views Bitcoin from a cooperative game theoretic perspective
and show under high transaction loads, it is difficult for pool managers to distribute rewards in a stable way;
\cite{chen2019axiomatic} adopts an axiomatic approach to investigate how the reward should be distributed among the miners.


\section{Model}

In this section, we formally define the {\em Discouraging pool block withholding games (DPBW games)}.
Similar with \cite{Eyal2015pbwa} and \cite{AlkalayHoulihan2019aaai},
it is assumed that each miner exactly joins one pool and is totally operated by that pool's manager.
Thus the players in a DPBW game are the managers of $n$ mining pools, denoted by $N$.
Let  $m_i \in \mathbb{R}^{+}$ be the mining power of manager $i\in N$.
Assume $m$ is the total mining power in the worldwide Bitcoin system and $m \ge \sum_{i \in N} m_i$.
By $m > \sum_{i \in N} m_i$, we mean there are
extra mining power outside of the studied pools $N$, due to any solo miners or inaccessible mining pools.

In a DPBW game, each player $i$ might only allocate $\alpha_{i}$ fraction of the total reward to her miners proportionally to their mining power
and award the left to the successful miner.
She might also use a fraction of its mining power to infiltrate another pool~$j$.
Such mining power does not actually work for pool $j$, but gets a fraction of the total reward from $j$.
It is assumed that the mining powers are continuous and can be arbitrarily divided.
In this work, we study how to select these $\alpha_{i}$'s for the players so that all of them do not want to infiltrate others.
Thus player $i$'s strategy space is all possible ${\bf x}_{i} = (x_{ij})_{j\in N\backslash\{i\}}$ such that $\sum_{j=1}^{n} x_{ij} \le m_{i}$ and $x_{ij} \ge 0$ for all $j\in N\backslash\{i\}$.
Each $x_{ij}$ with $j\neq i$ represents the amount of mining power that $i$ wants to infiltrate pool $j$.
Denote by ${\bf x} = ({\bf x}_{1}, \cdots, {\bf x}_{n})$ a full strategy profile.



To make us focus on the pool block withholding behaviors,
we assume the reward for each block is fixed,
thus the selection of the transactions does not matter.
Given a strategy profile ${\bf x}$, the players' utilities are defined as follows:
Assume each player $i$ first gets a total reward of $r_{i}({\bf x})$ by mining and infiltrating other pools.
Then she deducts $(1-\alpha_{i})$ fraction from $r_{i}({\bf x})$ used for award and obviously
only her honest miners can get this.
The remaining reward $\alpha_{i}r_{i}({\bf x})$ is proportionally allocated to all her miners
including both her honest and the infiltrated ones from the other pools.
Thus the pool's true utility is the total reward allocated to her honest miners.

More precisely, each player $i$'s reward $r_{i}({\bf x})$ consists of two parts:
{\em direct reward} and {\em infiltrating reward}.
Since every player $i$ only uses $m_{i} - \sum_{l \neq i} x_{il}$ mining power for honest mining,
player $i$'s direct reward is proportional to the fraction of the honest mining power contributed by her pool, denoted by
$$DR_{i}({\bf x}) = \frac{m_{i} - \sum_{l \neq i} x_{il}}{m -  \sum_{j=1}^{n} \sum_{l \neq j} x_{jl}};$$
since the pools cannot distinguish infiltrating miners from honest miners, player $i$'s infiltrating reward from every other pool $j$ is proportional to
the fraction of her infiltrating mining power to $j$, denoted by
$$IR_{i}({\bf x}) = \sum_{j\in N\setminus\{i\}} \frac{\alpha_{j} \cdot r_{j}({\bf x}) \cdot x_{i}} {m_{j} + \sum_{l\neq j}x_{lj}}.$$
Note that 
both the direct reward and the infiltrating reward will be allocated to all miners.
Accordingly, player $i$'s total reward is
$$r_{i}({\bf x}) = DR_{i}({\bf x}) + IR_{i}({\bf x}),$$
and her utility is

\begin{equation}\label{eq:utility}
    U_{i}({\bf x}) =  (1 - \alpha_{i} + \frac{m_{i}\alpha_{i}}{m_{i} + \sum_{j\neq i}x_{ji}}) r_{i}({\bf x}).
\end{equation}

For any vector ${\bf v} = (v_{1},\ldots, v_{n})$ and a particular $1 \le i \le n$,
denote by ${\bf v}_{-i}$ the resulting vector of ${\bf v}$ when element $v_{i}$ is omitted.
A strategy profile ${\bf x}$ is called a Nash equilibrium if
for every player $i$ and every possible strategy ${\bf x}'_{i}$,
$$
U_{i}({\bf x}) \ge U_{i}({\bf x}'_{i}, {\bf x}_{-i}).
$$

In a non-cooperative game, the {\em price of anarchy} (PoA) is defined as the ratio
between the optimal social welfare and the worst social welfare of any possible Nash equilibria;
while the {\em price of stability} (PoS) is defined as the ratio
between the optimal social welfare and the best social welfare of any possible Nash equilibria.
In DPBW games, the social welfare is defined as the total mining power that is used to honest mining.

\section{The PoS of DPBW Games is 1}

In this section, we prove that for most reasonable $\alpha_{i}$'s,
no-pool-attacks is always a Nash equilibrium of DPBW games for any number of mining pools, that is, the PoS of DPBW games is always 1. Formally,

\begin{theorem}\label{thm:anypools}
    Strategy profile $({\bf 0}, \ldots, {\bf 0})$ is a Nash equilibrium for any DPBW game
    if  $\forall i,  \alpha_{i} \leq  1-\frac{m_{\max}}{m}$, where $m_{\max}=\max_{i\in N} m_i$.
\end{theorem}

\begin{proof}
    Let ${\bf x}=({\bf x}_{1}, \ldots, {\bf x}_{n})$ be the strategy profile such that $x_{ij} = 0$ for any $i\in N$ and $j\in N\setminus\{i\}$.
    Note that, in Equation~\eqref{eq:utility}, the coefficient before $r_{i}({\bf x})$ is always 0 as $x_{ji} = 0$ for all $j\neq i$.
    Thus, in the rest proof we show $r_i({\bf x'}_i, {\bf x}_{-i})\ge r_i({\bf x})$ instead of $U_i({\bf x'}_i, {\bf x}_{-i})\ge U_i({\bf x}_i)$, where ${\bf x}'_i$ is any feasible strategy of player $i$.

    Arbitrarily fix a player $i$ and it is easy to see $r_i({\bf x}) = \frac{m_i}{m}$.
    To prove the theorem, it suffices to show that for any feasible strategy ${\bf x}_{i}'$ with all $j\neq i$ and $x'_{ij} \ge 0$,
    $$r_i({\bf x}'_{i}, {\bf x}_{-i}) \le \frac{m_i}{m}.$$
    Denote $i$'s total infiltrating mining power under strategy ${\bf x}'_{i}$ by $x_{i}' = \sum_{j\neq i}x'_{ij}$. Then
    \begin{eqnarray*}
        r_i({\bf x}'_{i}, {\bf x}_{-i}) &=& \frac{m_i-x'_i}{m-x'_i}+\sum_{j\in N\setminus\{i\}}\frac{\alpha_jx'_{ij}m_j}{(m-x'_i)(m_j+x'_{ij})}\\
                                        &=& \frac{m_i}{m}\left(\frac{m}{m_i}\cdot\frac{m_i-x'_i}{m-x'_i} +\sum_{j\in N\setminus\{i\}}\frac{m}{m_i}\cdot\frac{\alpha_jx'_{ij}m_j}{(m-x'_i)(m_j+x'_{ij})}\right)\\
                                        &=& \frac{m_i}{m}\left(1+\frac{x'_i}{m_i}\cdot\frac{m_i-m}{m-x'_i} +\sum_{j\in N\setminus\{i\}}\frac{m}{m_i}\cdot\frac{\alpha_jx'_{ij}m_j}{(m-x'_i)(m_j+x'_{ij})}\right)\\
                                        &=& \frac{m_i}{m}\left(1+\sum_{j\in N\setminus\{i\}}\frac{x'_{ij}}{m_i}\cdot\frac{m_i-m}{m-x'_i} +\sum_{j\in N\setminus\{i\}}\frac{m}{m_i}\cdot\frac{\alpha_jx'_{ij}m_j}{(m-x'_i)(m_j+x'_{ij})}\right).
    \end{eqnarray*}
    Now we claim the following inequality:
    \begin{equation}\label{eq:proof_of_thm1_1}
        \frac{x'_{ij}}{m_i}\cdot\frac{m_i-m}{m-x'_i}+
        \frac{m}{m_i}\cdot\frac{\alpha_jx'_{ij}m_j}{(m-x'_i)(m_j+x'_{ij})}\leq 0.
    \end{equation}
    Note that Inequality (\ref{eq:proof_of_thm1_1}) implies $r_i(\textbf{x}'_{i}, \textbf{x}_{-i}) \leq \frac{m_i}{m}=r_i(\textbf{x})$,
    which completes the proof of Theorem~\ref{thm:anypools}.

    The remaining is dedicated to show the correctness of Inequality (\ref{eq:proof_of_thm1_1}).
    By rearranging the terms, the lefthand of Inequality (\ref{eq:proof_of_thm1_1}) equals
    $$
    \frac{x'_{ij}[(m_i-m)(m_j+x'_{ij})+mm_j\alpha_{j}]}{m_i(m-x'_i)(m_j+x'_{ij})}.
    $$
    As the denominator is always positive,
    we only need to consider the sign of its numerator:
    \begin{eqnarray*}
        (m_i-m)(m_j+x'_{ij})+mm_j\alpha_{j} &\leq& (m_i-m)m_j+mm_j\alpha_{j}\\
                                            &\leq&  (m_i-m)m_j+mm_j(1-\frac{m_{max}}{m})\\
                                            &=&m_im_j-m_{max}m_j\leq 0,
    \end{eqnarray*}
    where the first inequality is because $x'_{ij} \ge 0$ and the second is because $\alpha_{i} \leq  1-\frac{m_{max}}{m}$.
\end{proof}

\paragraph{Remark.}
Note that $\frac{m_{\max}}{m}$ fraction of the reward cannot be too large for any dencentralized system,
thus we believe the requirement of $\alpha_{i} \leq  1-\frac{m_{\max}}{m}$ in Theorem~\ref{thm:anypools} is
a reasonable tradeoff between complementing the maintenance of a pool and incentivizing the miners to join it.

\section{The Uniqueness of Nash equilibrium for Two-Pool Case}
In this section, we study a special case of the DPBW game when only two pools are included,
which is exactly the same setting with the previous work \cite{AlkalayHoulihan2019aaai}.
However, as will be proved, in PBWA+ game, it is possible for the pool managers to deduct a small fraction from the reward,
so that no-pool-attacks is a unique Nash equilibrium.

\begin{theorem}
    \label{thm:2pools:unique}
    For PBWA+ Game with two players, by setting $\alpha_1 \le 1-\frac{m_2}{m}$, $\alpha_2 \le 1-\frac{m_1}{m}$ and $m > 3(m_1 + m_2)$,
    the game has a unique Nash equilibrium where both players do not infiltrate the other pool.
\end{theorem}

We believe the requirement of $m > 3(m_1 + m_2)$ in the theorem is reasonable as
the statistic website\footnote{\url{https://btc.com/stats/pool?pool_mode=month}}
shows that the largest two pools have roughly a third of the total computational power.

\subsection{Notations and Proof Ideas of Theorem \ref{thm:2pools:unique}}

Before we prove Theorem \ref{thm:2pools:unique},
we first simplify our notions to ease our representation.
Since there are only two players, we simplify our notions as follows.
Let  $(x_{1}, x_{2})$ be a strategy profile, where $x_i \in [0, m_{i}]$ means how much player $i$ infiltrates player $j = 3-i$.
Thus each player $i$'s direct reward $r_{i}(x_{1}, x_{2})$ is
$$DR_{i}(x_{1}, x_{2})=\frac{m_{i} - x_{i}}{m - x_1 - x_2};$$
her total reward is
\begin{eqnarray*}
    r_{i}(x_{1}, x_{2}) & = &
    R_{i}(x_{1}, x_{2})  + \frac{\alpha_{j} r_{j}(x_1, x_2) x_{i}} {m_{j} + x_{i}} =  \frac{m_{i} - x_{i}}{m - x_1 - x_2} + \frac{\alpha_{j} r_{j}(x_1, x_2) x_{i}} {m_{j} + x_{i}};  
\end{eqnarray*}
and her utility is
$$
U_{i}(x_{1}, x_{2}) =  (1 - \alpha_{i} + \frac{m_{i}\alpha_{i}}{m_{i} + x_{j}}) r_{i}(x_{1}, x_{2}).
$$
Here note that in $U_{i}(x_{1}, x_{2})$, the first coefficient $1 - \alpha_{i} + \frac{m_{i}\alpha_{i}}{m_{i} + x_{j}}$ only depends on $x_{j}$,
thus when we analyze player $i$'s utility gain by unilateral deviation,
without loss of generality, we just ignore this coefficient and only consider her total reward $r_{i}(x_{1}, x_{2})$.

Solving the reward system of equations in $r_{i}(x_{1}, x_{2})$ for $i=1,2$,
we get the closed forms for the reward functions:
\begin{eqnarray*}
& r_{1}(x_{1}, x_{2}) =  \dfrac{(m_1 + x_2) \left[(m_1 - x_1)(m_2 + x_1) + \alpha_2 x_1 (m_2 - x_2) \right]}
{(m - x_1 - x_2) \left[(m_1 + x_2)(m_2 + x_1) - \alpha_1 \alpha_2 x_1 x_2 \right]}
\end{eqnarray*}
and
\begin{eqnarray*}
& r_{2}(x_{1}, x_{2}) = \dfrac{(m_2 + x_1) \left[(m_2 - x_2)(m_1 + x_2) + \alpha_1 x_2 (m_1 - x_1) \right]}
{(m - x_2 - x_1) \left[(m_2 + x_1)(m_1 + x_2) - \alpha_1 \alpha_2 x_1 x_2 \right]}.
\end{eqnarray*}

Let $x_i+ \Delta x_i$ be player $i$'s deviation, where $\Delta x_i \in [-x_i, m_i - x_i]$.
Denote by
$$ f_1(\Delta x_1)  =  r_1(x_1+ \Delta x_1, x_2) - r_1(x_1, x_2)$$
and
$$ f_2(\Delta x_2)  =  r_2(x_1, x_2 + \Delta x_2) - r_2(x_1, x_2)$$
the reward gain by unilateral deviation.
By definition, the necessary and sufficient condition for strategy profile $(x_1, x_2)$ to be a Nash equilibrium is
for any $\Delta x_i \in [-x_i, m_i - x_i]$,
$$
\begin{array}{ll}
    \left\{ \begin{array}{ll}
            f_1(\Delta x_1)  \leq 0 \\
            f_2(\Delta x_2)  \leq 0. \\
    \end{array} \right.
\end{array}
$$
Thus the problem of finding all possible Nash equilibria becomes finding all such $(x_1, x_2)$ strategy profiles.

Since players 1 and 2 are symmetric, in the following, without loss of generality, we often use player 1 for illustration.
The formula of $f_1(\Delta x_1)$ can be simplified as a quadratic function of $\Delta x_1$:
\begin{eqnarray}
&f_1(\Delta x_1) = A_1(x_1, x_2)  (\Delta x_1)^2 + B_1(x_1, x_2)  \Delta x_1 + C_1(x_1, x_2), \nonumber
\end{eqnarray}
where the formulas of $A_1(x_1, x_2)$, $B_1(x_1, x_2)$ and $C_1(x_1, x_2)$ are shown in Table~\ref{table:eq:ABC}.

\begin{table*}[h]
    \begin{center}
        \begin{eqnarray*}
            A_1(x_1, x_2) & = &- (m_1 + x_2) + r_1(x_1, x_2) (m_1 + x_2 - \alpha_1 \alpha_2 x_2) \\
            B_1(x_1, x_2) & = &(m_1 + x_2) \left[m_1 - m_2 - 2x_1 + \alpha_2(m_2 - x_2) \right] \\
                          &&+ \dfrac{(m_1 + m_2) \left[(m_1 - x_1)(m_2 + x_1) + \alpha_2 x_1 (m_2 - x_2)   \right]}{m - x_1 - x_2} \\
                          & & + \dfrac{(m_1 + x_2 - \alpha_1 \alpha_2 x_2) (m_1 + x_2) \left[(m_1 - x_1)(m_2 + x_1) + \alpha_2 x_2 (m_2 - x_2))  \right]}
                          {(m_1 + x_2) (m_2 + x_1) - \alpha_1 \alpha_2 x_1 x_2},\\
            C_1(x_1, x_2) & = & 0.
        \end{eqnarray*}
    \end{center}
    \caption{The formulas of $A_1(x_1, x_2)$, $B_1(x_1, x_2)$ and $C_1(x_1, x_2)$.}
    \label{table:eq:ABC}
\end{table*}%

Note that since $r_1(x_1, x_2) \leq 1$,  $A_1(x_1, x_2) \le 0$.
Thus, for any strategy profile $(x_1, x_2)$, in order to make $f_{1}(\Delta x_1) \le 0$ for any $\Delta x_1 \in [-x_1, m_1 - x_1]$,
there are three possible cases, as shown in Figures \ref{figure:strategy:x1=0}, \ref{figure:strategy:0<x1<m1} and \ref{figure:strategy:x1=m1}:
\begin{description}
    \item[Case 1. ] $x_1 = 0$ and $B_1(x_1, x_2) \le 0$;
    \item[Case 2. ] $0  < x_1 < m_1$ and $B_1(x_1, x_2) = 0$; and
    \item[Case 3. ] $x_1 = m_1$ and $B_{1}(x_1, x_2) \ge 0$.
\end{description}
Symmetrically, we have all the corresponding definitions of $f_2(\Delta x_2)$, $A_2(x_1, x_2)$, $B_2(x_1, x_2)$ and $C_2(x_1, x_2)$ for player~2.

Therefore, combining with players 1's and 2's strategies, we have nine kinds of possible Nash equilibria.
By Theorem~\ref{thm:anypools}, we already know that $(0,0)$ is a Nash equilibrium.
In the following, we will show that $(0,0)$ is actually the only possible Nash equilibrium for among the nine situations.
Thus in any DPBW game with two players, both the PoA and the PoS are 1.

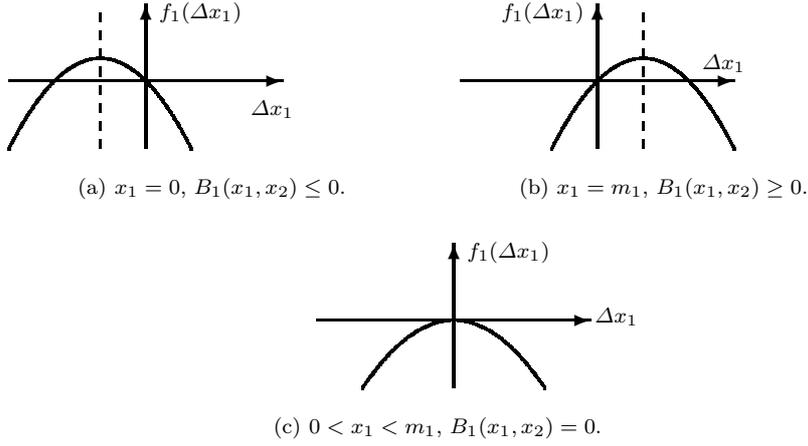
\begin{figure}[ht]
    \begin{subfigure}{0.5\linewidth} 
        \thicklines
        \setlength{\unitlength}{0.6cm}
        \begin{picture}(4.3,3.6)(-3.5,-0.25)
            \put(-3,1.5){\vector(1,0){6}}
            \put(2.3,.75){$\Delta x_1$}
            \put(0,0){\vector(0,1){3.2}}
            \put(1.2,3){\makebox(0,0){$f_1(\Delta x_1)$}}

            \multiput(-1,0)(0, 0.4){8}{\line(0,1){0.2}}

            \qbezier(-1, 2)(0, 2) (1, 0)
            \qbezier(-1, 2)(-2, 2) (-3, 0)
        \end{picture}
        \caption{$x_1 = 0$, $B_1(x_1, x_2) \le 0$.}
        \label{figure:strategy:x1=0}
    \end{subfigure}%
    \begin{subfigure}{0.5\linewidth} 
        \thicklines
        \setlength{\unitlength}{0.6cm}
        \begin{picture}(4.3,3.6)(-3.5,-0.25)
            \put(-3,1.5){\vector(1,0){6}}
            \put(2.3,1.75){$\Delta x_1$}
            \put(0,0){\vector(0,1){3.2}}
            \put(-1.2,3){\makebox(0,0){$f_1(\Delta x_1)$}}

            \qbezier(1, 2)(2, 2) (3, 0)
            \qbezier(1, 2)(0, 2) (-1, 0)
            \multiput(1,0)(0, 0.4){8}{\line(0,1){0.2}}

        \end{picture}
        \caption{$x_1 = m_1$, $B_{1}(x_1, x_2) \ge 0$.}
        \label{figure:strategy:x1=m1}
    \end{subfigure}%
    \newline

    \vspace{0.4cm}
    \begin{subfigure}{1\linewidth} 
        \centering
        \thicklines
        \setlength{\unitlength}{0.6cm}
        \begin{picture}(4.3,3.6)(-2.5,-0.25)
            \put(-3,1.5){\vector(1,0){6}}
            \put(3.1,1.45){$\Delta x_1$}
            \put(0,0){\vector(0,1){3.2}}
            \put(1.2,3){\makebox(0,0){$f_1(\Delta x_1)$}}

            \qbezier(0, 1.5)(1, 1.5) (2, 0)
            \qbezier(0, 1.5)(-1, 1.5) (-2, 0)
        \end{picture}
        \caption{$0  < x_1 < m_1$, $B_1(x_1, x_2) = 0$.}
        \label{figure:strategy:0<x1<m1}
    \end{subfigure}
    \caption{Three possible strategies for player $1$}
    \label{figure:strategy}
\end{figure}

\subsection{Proof of Theorem \ref{thm:2pools:unique}}

In the remaining of this section, we prove Theorem~\ref{thm:2pools:unique} by
showing that all the eight situations, except $(0, 0)$, are not possible
to be Nash equilibria.

\subsection{Case 2 + Case 2: $B_1(x_1, x_2) = B_2(x_1, x_2) = 0$}

To simplify our analysis, we first multiply
$B_1(x_1, x_2)$ by $$\dfrac{(m - x_1 - x_2)\left[(m_1 + m_2)(m_2 + m_1) -\alpha_1 \alpha_2 x_1 x_2) \right]}{m_1 + x_2}$$
to eliminate the denominator, denoting the resulting polynomial by $Q_1(x_1, x_2)$.
By rearranging every monomials, $Q_1(x_1, x_2)$ can be denoted as a quadratic function of $x_{1}$:
$$
Q_1(x_1, x_2) =  a_1(x_2) x_1^2 + b_1(x_2) x_1 + c_1(x_2),
$$
where $a_1(x_2)$, $b_1(x_2)$ and $c_1(x_2)$ are
independent of $x_{1}$ and their formulas are shown in Table~\ref{table:eq:abcsmall}.
Similarly, we also have all the corresponding definitions of $Q_2(x_1, x_2), a_2(x_1), b_2(x_1)$, and $c_2(x_1)$ for player 2.

\begin{table*}[h]
    \begin{center}
        \begin{eqnarray*}
            a_1(x_2)& = &\left(-\alpha_1 \alpha_2^2 m_2+\alpha_1 \alpha_2 m-\alpha_1 \alpha_2 m_1-\alpha_2 m_1+\alpha_1 \alpha_2 m_2+\alpha_2 m_2-m+2 m_1\right)x_2 \\
                    &&+\left(\alpha_1 \alpha_2^2-\alpha_1 \alpha_2-\alpha_2+1\right) x_2^2+\alpha_2 m_2 m_1+m_1^2-m m_1  ,\\
            b_1(x_2) & = & -2m m_1 m_2 + 2 m_1^2 m_2 + (-2 m m_2 + 4 m_1 m_2 - 2 \alpha_1 \alpha_2 m_1 m_2) x_2 + 2 m_2 x_2^2 ,\\
            c_1(x_2) & = & \alpha_2 m_2 x_2^3 +(-m\alpha_2 m_2 + \alpha_2 m_1 m_2 - \alpha_1 \alpha_2 m_1 m_2 + m_2^2 - \alpha_2 m_2^2) x_2^2 \\
                     &&+(-mm_1m_2^2 + m\alpha_2m_1m_2^2 + m_1m_2^2) \\
                     &&+(-m\alpha_2m_1m_2 + m\alpha_1\alpha_2m_1m_2 - m m_2^2 + m\alpha_2m_2^2 + 2 m_1m_2^2 - \alpha_2m_1m_2^2) x_2.
        \end{eqnarray*}
        \caption{The formulas of $a_1(x_1, x_2)$, $b_1(x_1, x_2)$ and $c_1(x_1, x_2)$.}
        \label{table:eq:abcsmall}
    \end{center}
\end{table*}
To show the impossibility of Case 1, it suffices to prove the following lemma.
\begin{lemma}\label{lem:q1=q2=0}
    For any strategy profile $(x_1, x_2)$, $Q_1(x_1, x_2)$ and $Q_2(x_1, x_2)$ cannot be 0 simultaneously.
\end{lemma}

Before we prove Lemma \ref{lem:q1=q2=0}, we first prove the following claims (Claim \ref{claim:Q_1:c_1}, Claim \ref{claim:Q_1:b1} and Claim \ref{claim:Q_1:a1<0}).
\begin{claim}
    \label{claim:Q_1:c_1}
    $c_1(x_2) \le 0$ for any $x_2\in [0, m_2]$.
\end{claim}

\begin{proof}
    To prove Claim \ref{claim:Q_1:c_1}, We first show the following facts:
    \begin{itemize}
        \item $c_1(0) \le 0$;
        \item $c_1(m_2) <  0$;
        \item $c_1(x_2)$ achieves its local maximal value at some $y<0$.
    \end{itemize}
    If the above three facts setting up, the curve of $c_1(x_2)$ can be shown in Figure \ref{figure:c1},
    and $c_1(y) < 0$ for any $y\in [0, m_2]$ (the red part),

    Combing the three facts, the curve of $c_1(x_2)$ can be shown in Figure \ref{figure:c1},
    and $c_1(x_2) < 0$ for any $x_2\in [0, m_2]$ (the red part),
    which completes the proof of the claim.

    \begin{figure}[htbp]
        \begin{center}
            \setlength{\unitlength}{0.6cm}
            \thicklines
            \setlength{\unitlength}{0.6cm}
            \begin{picture}(4.3,3.6)(-2.5,-0.25)
                \put(-3,1.5){\vector(1,0){6}}
                \put(3.1,1.45){$x_2$}
                \put(0,0){\vector(0,1){3.2}}
                \put(-0.8,3.3){\makebox(0,0){$c_1(y)$}}

                {\color{red}\qbezier(0, 0.7)(0.9, -1) (2, 1.5)}
                \qbezier(2, 1.5)(2.15,1.8) (2.5, 2.8)
                \qbezier(0, 0.7)(-1.5, 4) (-2.8, 0.2)

                \multiput(-1.35,0)(0, 0.4){8}{\line(0,1){0.2}}

                \multiput(0.6,0)(0, 0.4){8}{\line(0,1){0.2}}

            \end{picture}
            \caption{$c_1(x_2)$.}
            \label{figure:c1}
        \end{center}
    \end{figure}

    Now we prove these facts one by one.

    To see the first fact, observe that when $\alpha_2 \le 1- \frac{m_1}{m}$,
    \[
        c_1(0) = (\alpha_2 - 1) m m_1 m_2^2 + m_1^2 m_2^2  \le -\frac{m_1}{m} m m_1 m_2^2 + m_1^2 m_2^2 \le 0.
    \]
    To see the second fact, we note that
    \begin{eqnarray*}
        \frac{c_1(m_2)}{m_2^2} & = & m_1^2 -m m_2 +m_2^2 -m m_1 + \alpha_1 \alpha_2 m m_1 - \alpha_1 \alpha_2 m_1 m_2 + 2 m_1 m_2\\
                               & = & m_1^2 + m_2^2 - m m_2 + (\alpha_1 \alpha_2 - 1)mm_1   + (2 - \alpha_1 \alpha_2)m_1m_2\\
                               & < & m_1^2 + m_2^2 - m m_2 + (\alpha_2 - 1) m m_1 + 2 m_1 m_2 \\
                               & \le & m_1^2 + m_2^2 - 2 m_1 m_2 - m_2^2 - m_1^2 + 2 m_1 m_2 \\
                               &=& 0,
    \end{eqnarray*}
    where the last inequality is because $m \ge 2(m_1 + m_2)$ and $\alpha_2 - 1 \le - \frac{m_1}{m}$.

    To prove the third fact, we first compute the derivative of $c_1(x_2)$,
    \begin{eqnarray*}
        c'_1(x_2) & = & (3 \alpha_2 m_2) x_2^2 - 2( \alpha_2 m m_2 - \alpha_2 m_1 m_2 + \alpha_1 \alpha_2 m_1 m_2 - m_2^2 + \alpha_2 m_2^2) x_2 \\
                  && -\alpha_2 m m_1 m_2 + \alpha_1 \alpha_2 m_1 m_2 - m m_2^2  + \alpha_2 m m_2^2 + 2 m_1 m_2^2 - \alpha_2 m_1 m_2^2,
    \end{eqnarray*}
    which is a quadratic function of $x_2$.
    Combing the previous two facts and the property that $3 \alpha_2 m_2 > 0$, to show the third fact, it suffices to show the ``constant term'' of $c'_1(x_2)$ is smaller than 0.
    This sufficient condition can be proved by the following derivation.

    \begin{eqnarray*}
 & & - \alpha_2 m m_1 m_2 + \alpha_1 \alpha_2 m_1 m_2 - m m_2^2  + \alpha_2 m m_2^2 + 2 m_1 m_2^2 - \alpha_2 m_1 m_2^2\\
 && =  - \alpha_2 m m_1 m_2 (1 - \alpha_1) - m m_2^2 (1 - \alpha_2) + m_1 m_2^2 (2 - \alpha_2) \\
 && \le  -\alpha_2 m_1 m_2^2 - m_1 m_2^2 + m_1 m_2^2 (2 - \alpha_2) \\
 && =  m_1 m_2 (-\alpha_2 - 1 + 2 - \alpha_2) = m_1 m_2 (1 - 2 \alpha_2) < 0.
    \end{eqnarray*}
    Again, the first inequality is because $1 - \alpha_1 \ge \frac{m_2}{m}$ and $1 - \alpha_2 \ge \frac{m_1}{m}$;
    and the second is because $\alpha_i > \frac{1}{2}$.
\end{proof}
\begin{claim}
    \label{claim:Q_1:b1}
    $b_1(x_2) < 0$ for any $x_2 \in [0, m_2]$.
\end{claim}
\begin{proof}
    As $m \geq 3(m_1 + m_2)$, the claim easily follows by the following inequalities.
    \begin{eqnarray*}
        b_1(x_2)
& = & -2m m_1 m_2 + 2 m_1^2 m_2  + 2 m_2 x_2^2 + (-2 m m_2 + 4 m_1 m_2 - 2 \alpha_1 \alpha_2 m_1 m_2) x_2 \\
&<&  -2m_1^2 m_2 + 2 m_1^2 m_2  + 2 m_2 x_2^2  + (-2 (2(m_1+x_2)) m_2 + 4 m_1 m_2 ) x_2 < 0.
    \end{eqnarray*}
\end{proof}


\begin{claim}
    \label{claim:Q_1:a1<0}
    (1) Given any $0 < x_2 < m_2$ such that $a_1(x_2) \le 0$, then $Q_1(x_1, x_2) \neq 0$ for all $0 < x_1 < m_1$;
    (2) given any $0 < x_1 < m_1$ such that $a_2(x_1) \le 0$, then $Q_2(x_1, x_2) \neq 0$ for all $0 < x_2 < m_2$.
\end{claim}
\begin{proof}
    Since players 1 and 2 are symmetric, we only prove for player 1.
    If $a_1(x_2)=0$, then the curve of $Q(x_1, x_2)$ (a quadratic function of $x_1$) will degenerate to a line (see Figure~\ref{figure:a1=0}). Since its slope $b_1(x_2)$ is negative and y-intercept is non-positive, $Q_1(x_1, x_2)$ cannot be $0$ for any $x_1>0$.

    \begin{figure}[ht]
        \begin{subfigure}{0.5\linewidth}
            \setlength{\unitlength}{0.6cm}
            \thicklines
            \setlength{\unitlength}{0.6cm}
            \begin{picture}(4.3,3.6)(-3.5,-0.25)
                \put(-3,1.5){\vector(1,0){6}}
                \put(3.1,1.45){$x_1$}
                \put(0,0){\vector(0,1){3.2}}
                \put(1.5,3){\makebox(0,0){$Q_1(x_1, x_2)$}}

                \put(0,1.3){\line(-2,1){2}}
                {\color{red}\put(0,1.3){\line(2,-1){2}}}
            \end{picture}
            \caption{$a_1(x_2)=0$}
            \label{figure:a1=0}
        \end{subfigure}%
        \begin{subfigure}{0.5\linewidth}
            \setlength{\unitlength}{0.6cm}
            \thicklines
            \setlength{\unitlength}{0.6cm}
            \begin{picture}(4.3,3.6)(-3.5,-0.25)
                \put(-3,1.5){\vector(1,0){6}}
                \put(3.1,1.45){$x_1$}
                \put(0,0){\vector(0,1){3.2}}
                \put(1.5,3){\makebox(0,0){$Q_1(x_1, x_2)$}}

                \qbezier(-2.8, 0.5)(-1.2, 3.8) (0, 1.15)
                {\color{red} \qbezier(0,1.15)(0.18, 0.75) (0.4, 0.3)}

            \end{picture}
            \caption{$a_1(x_2)<0$}
            \label{figure:a1<0}
        \end{subfigure}%
        \caption{$Q_1$: The case when $a_1(x_{2}) \le 0$.}
        \label{figure:a1<=0}
    \end{figure}
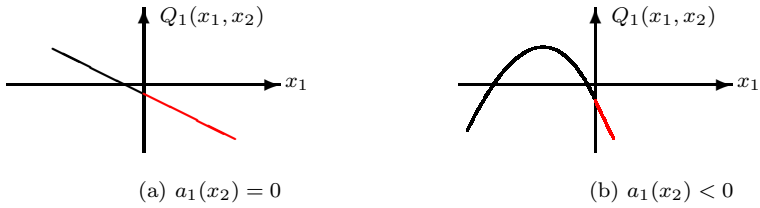
    As $b_1(x_2) < 0$ and $a_1(x_2) < 0$,
    $$-\dfrac{b_1(x_2)}{2a_1(x_2)} < 0.$$
    Combing with the result $c_1(x_2)\leq 0$ for any $x_2\in[0,m_2]$ in Claim \ref{claim:Q_1:c_1}, the curve of $Q_1$ (a quadratic function of $x_{1}$) can only be the case shown in Figure~\ref{figure:a1<0}.
    Accordingly, $Q_1(x_1, x_2) < 0$ for all $0 < x_1 < m_1$ (the red part).
\end{proof}

Having the above three Claims, we are ready to show the proof of Lemma~\ref{lem:q1=q2=0}.

\subsubsection{Proof of Lemma \ref{lem:q1=q2=0}.}
We prove by contradiction.
If there is a strategy profile $(x_1, x_2)$, $0 < x_1 < m_1$ and $0 < x_2 < m_2$, such that $Q_1(x_1, x_2) = Q_2(x_1, x_2) = 0$,
by Claim \ref{claim:Q_1:a1<0}, it must be the case that (shown in Figure \ref{figure:a1>0})
\begin{enumerate}
    \item $a_1(x_2) > 0$ and $Q_1(m_1, x_2) > 0$; and
    \item $a_2(x_1) > 0$ and $Q_2(x_1, m_2) > 0$.
\end{enumerate}

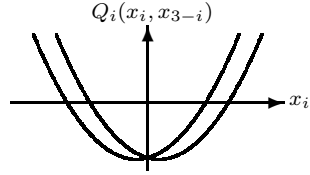
\begin{figure}[htbp]
    \begin{center}
        \setlength{\unitlength}{0.6cm}
        \thicklines
        \setlength{\unitlength}{0.6cm}
        \begin{picture}(4.3,3.6)(-2.5,-0.25)
            \put(-3,1.5){\vector(1,0){6}}
            \put(3.1,1.45){$x_i$}
            \put(0,0){\vector(0,1){3.2}}
            \put(0.1,3.5){\makebox(0,0){$Q_i(x_i, x_{3-i})$}}

            \qbezier(-2, 3)(0.25, -2.5) (2.5, 3)
            \qbezier(2, 3)(-0.25, -2.5) (-2.5, 3)

        \end{picture}
        \caption{$Q_i$: if $a_i(x_{3-i}) \ge 0$.}
        \label{figure:a1>0}
    \end{center}
\end{figure}

Note that
\begin{eqnarray}
    &&Q_1(m_1, x_2)  = \nonumber\\
    && -\alpha_2 m_1 (m - m_1 - x_2) (m_2 - x_2)(m_1 + x_2 - \alpha_1 \alpha_2 x_2) \label{eq:q1:1}   \\
    && + (-m_1 - m_2 + \alpha_2 (m_2 - x_2)) (m - m_1 - x_2)  \nonumber \\
    &&  ~~~\cdot ( - \alpha_1 \alpha_2 m_1 x_2 + (m_1 + m_2) (m_1 + x_2)) \label{eq:q1:2}  \\
    && + \alpha_2 m_1 (m_2 - x_2) (-\alpha_1 \alpha_2 m_1 x_2 + (m_1 + m_2) (m_1 + x_2)). \label{eq:q1:3}
\end{eqnarray}
Let $\bar{Q}_1 (x_2) = \frac{(\ref{eq:q1:1}) + (\ref{eq:q1:3})}{\alpha_2m_1(m_2-x_2)}$, then $\bar{Q}_1 (x_2) > 0$ since (\ref{eq:q1:2}) is always negative and $\alpha_2m_1(m_2-x_2)$ is positive.
Thus,
\begin{eqnarray*}
    \bar{Q}_1(x_2) & = & 2m_1^2 + x_2 (-m + m_2 + \alpha_1 \alpha_2 (m - x_2) + x_2) \\
                   && + m_1 (-m + m_2 + (3 - 2 \alpha_1 \alpha_2) x_2).
\end{eqnarray*}
We assume $m=\ell(m_1 + m_2)$ for some $\ell\geq 3$, thus
\begin{eqnarray*}
    \bar{Q}_1(m_2) & = & (2 - \alpha) m_1^2 + (2- \alpha - \alpha_1 \alpha_2 + \alpha \alpha_1 \alpha_2) m_2^2 \\
                   && + (4 - 2 \ell + 2 \alpha_1 \alpha_2(-2 + \ell)) m_1 m_2 > 0.
\end{eqnarray*}
By the symmetricity of $\bar{Q}_1(m_2)$ and $\bar{Q}_2(m_1)$,  we have
\[
    \bar{Q}_1(m_2) + \bar{Q}_2(m_1)  =  (4 - 2 \ell - \alpha_1 \alpha_2 + \ell \alpha_1 \alpha_2) (m_1 + m_2)^2 - 2 \alpha_1 \alpha_2 m_1 m_2 > 0,
\]
which implies
$$
\alpha < \dfrac{4 - \alpha_1 \alpha_2}{2 - \alpha_1 \alpha_2} \le 3.
$$
Thus we get a contradiction with $m \geq 3(m_1 + m_2)$,
which completes the proof of Lemma \ref{lem:q1=q2=0} and accordingly Case 1 + Case 1 cannot be a Nash equilibrium.

\subsection{Case 2 + Case 1/3: $B_1(x_1, x_2) = 0$ and $x_2\in\{0, m_2\}$}
We follow the notation $Q_1(\cdot)$ in the prior subsection,
\[
    Q_1(x_1, x_2) = a_1(x_2)x_1^2+b_1(x_2)x_1+c_1(x_2)
\]
and $c_1(x_{2}) \le 0$ has been shown in Claim~\ref{claim:Q_1:c_1}.
\begin{lemma} \label{lem:case2+case13}
    $Q_1(x_1, x_2)$ cannot be $0$ for any $x_1\in (0, m_1)$ and $x_2\in\{0, m_2\}$.
\end{lemma}
\begin{proof}
    As shown in Claim~\ref{claim:Q_1:c_1} and Claim~\ref{claim:Q_1:b1}, we have $c_1(x_2) \le 0$ and $b_1(x_2) < 0$ for any $x_2\in\{0, m_2\}$.
    Similar to figure~\ref{figure:a1=0} and \ref{figure:a1<0}, if $a_1(0) \le 0$ then $Q_1(x_{1}, 0)$ cannot be $0$ for all $x_1\ge0$.
    And for the case $a_1(0) \ge 0$,
    \begin{eqnarray*}
        Q_1(m_1, 0) & = & \alpha_2m_1^2m_2(-m+2m_1+m_2) \\
        Q_1(m_1, m_2) & = & -(m_1+m_2)(m-m_1-m_2) (-\alpha_1\alpha_2m_1m_2+(m_1+m_2)^2)
    \end{eqnarray*}
    are obviously less than $0$.
    Thus, $Q_1$ cannot be $0$ for any $x_1\in(0,m_1)$ and $x_2\in\{0, m_2\}$.
\end{proof}
By Lemma \ref{lem:case2+case13}, Case 2 + Case 1/3 cannot be Nash equilibria,
and symmetrically, Case 1/3 + Case 2 cannot either.

\subsection{Other Situations}
Finally, we briefly discuss the remaining situations: $(x_1,x_2)\in \{(0, m_2), (m_1, 0),\allowbreak (m_1, m_2)\}$.
It is not hard to see that neither of them can benefit by sending all her mining power to attack the other:
for case $(m_1, m_2)$, both two players get no reward;
for case $(0, m_2)$, player $2$ only gets
$$\frac{m_1}{m-m_2}\cdot\frac{m_2}{m_1+m_2}(1-\alpha_1) < \frac{m_2}{m};$$
and the case $(m_1, 0)$ is symmetric to $(0, m_2)$.

In conclusion, all the eight situations, except $(0, 0)$, are not Nash equilibria,
which finishes the proof of Theorem~\ref{thm:2pools:unique}.

\section{Conclusion and Future Directions}
In this work, we refine the game-theoretic model for pool block withholding attacks when the managers of the pools
individually deduct small fractions from the total rewards.
For most reasonable  deductions,
we show that no-pool-attacks is always a Nash equilibrium for any number of pools, and particularly,
when there are only two pools under consideration, it is the unique Nash equilibrium.

A direct open problem is to generalize our second result to more than two pools, that is, is it possible for an arbitrary number of managers to make small deductions
so that no-pool-attacks is also a unique Nash equilibrium?

There are many other future directions that deserve exploration.
For example, in all works about PBW games that we know, it is assumed that
every miner for a pool is always loyal to her pool manager.
This assumption does not hold when the miners have their own interests on how to report their solutions.
We hope that, by combining tools in cooperative game theory, our work can inspire further study in this direction.


%
%


\bibliographystyle{plain}
\bibliography{pbwa}

\begin{thebibliography}{10}

\bibitem{AlkalayHoulihan2019aaai}
Colleen Alkalay-Houlihan and Nisarg Shah.
\newblock The pure price of anarchy of pool block withholding attacks in
  bitcoin mining.
\newblock In {\em AAAI 2019}, 2019.

\bibitem{babaioff2012bitcoin}
Moshe Babaioff, Shahar Dobzinski, Sigal Oren, and Aviv Zohar.
\newblock On bitcoin and red balloons.
\newblock In {\em Proceedings of the 13th ACM conference on electronic
  commerce}, pages 56--73. ACM, 2012.

\bibitem{Bag2017}
S.~{Bag}, S.~{Ruj}, and K.~{Sakurai}.
\newblock Bitcoin block withholding attack: Analysis and mitigation.
\newblock {\em IEEE Transactions on Information Forensics and Security},
  12(8):1967--1978, 2017.

\bibitem{Bag2016preventBWA}
Samiran Bag and Kouichi Sakurai.
\newblock Yet another note on block withholding attack on bitcoin mining pools.
\newblock In Matt Bishop and Anderson C~A Nascimento, editors, {\em Information
  Security}, pages 167--180, Cham, 2016. Springer International Publishing.

\bibitem{carlsten2016instability}
Miles Carlsten, Harry Kalodner, S~Matthew Weinberg, and Arvind Narayanan.
\newblock On the instability of bitcoin without the block reward.
\newblock In {\em Proceedings of the 2016 ACM SIGSAC Conference on Computer and
  Communications Security}, pages 154--167. ACM, 2016.

\bibitem{chen2019axiomatic}
Xi~Chen, Christos Papadimitriou, and Tim Roughgarden.
\newblock An axiomatic approach to block rewards.
\newblock In {\em Proceedings of the 1st ACM Conference on Advances in
  Financial Technologies}, pages 124--131. ACM, 2019.

\bibitem{Eyal2015pbwa}
I.~{Eyal}.
\newblock The miner's dilemma.
\newblock In {\em 2015 IEEE Symposium on Security and Privacy}, pages 89--103,
  May 2015.

\bibitem{Eyal2014selfishmining}
Ittay Eyal and Emin~G{\"u}n Sirer.
\newblock Majority is not enough: Bitcoin mining is vulnerable.
\newblock In Nicolas Christin and Reihaneh Safavi-Naini, editors, {\em
  Financial Cryptography and Data Security}, pages 436--454, Berlin,
  Heidelberg, 2014. Springer Berlin Heidelberg.

\bibitem{Haghighat2019}
Alireza~Toroghi Haghighat and Mehdi Shajari.
\newblock Block withholding game among bitcoin mining pools.
\newblock {\em Future Generation Computer Systems}, 97:482 -- 491, 2019.

\bibitem{Heilman2015129eclipse}
E.~Heilman, A.~Kendler, A.~Zohar, and S.~Goldberg.
\newblock Eclipse attacks on bitcoin's peer-to-peer network.
\newblock {\em Eclipse Attacks on Bitcoin's Peer-to-peer Network},
  (3):129--144, 2015.
\newblock cited By 106.

\bibitem{Johnson2014ddos}
Benjamin Johnson, Aron Laszka, Jens Grossklags, Marie Vasek, and Tyler Moore.
\newblock Game-theoretic analysis of ddos attacks against bitcoin mining pools.
\newblock In Rainer B{\"o}hme, Michael Brenner, Tyler Moore, and Matthew Smith,
  editors, {\em Financial Cryptography and Data Security}, pages 72--86,
  Berlin, Heidelberg, 2014. Springer Berlin Heidelberg.

\bibitem{lewenberg2015bitcoin}
Yoad Lewenberg, Yoram Bachrach, Yonatan Sompolinsky, Aviv Zohar, and Jeffrey~S
  Rosenschein.
\newblock Bitcoin mining pools: A cooperative game theoretic analysis.
\newblock In {\em Proceedings of the 2015 International Conference on
  Autonomous Agents and Multiagent Systems}, pages 919--927. Citeseer, 2015.

\bibitem{luu2015power}
Loi Luu, Ratul Saha, Inian Parameshwaran, Prateek Saxena, and Aquinas Hobor.
\newblock On power splitting games in distributed computation: The case of
  bitcoin pooled mining.
\newblock In {\em 2015 IEEE 28th Computer Security Foundations Symposium},
  pages 397--411. IEEE, 2015.

\bibitem{Mousavinejad2018detect}
E.~{Mousavinejad}, F.~{Yang}, Q.~{Han}, and L.~{Vlacic}.
\newblock A novel cyber attack detection method in networked control systems.
\newblock {\em IEEE Transactions on Cybernetics}, 48(11):3254--3264, Nov 2018.

\bibitem{bitcoin2008}
Satoshi Nakamoto.
\newblock Bitcoin: A peer-to-peer electronic cash system.
\newblock 2008.

\bibitem{Rosenfeld2011bwa}
Meni Rosenfeld.
\newblock Analysis of bitcoin pooled mining reward systems.
\newblock {\em CoRR}, abs/1112.4980, 2011.

\bibitem{schrijvers2016incentive}
Okke Schrijvers, Joseph Bonneau, Dan Boneh, and Tim Roughgarden.
\newblock Incentive compatibility of bitcoin mining pool reward functions.
\newblock In {\em International Conference on Financial Cryptography and Data
  Security}, pages 477--498. Springer, 2016.

\bibitem{Tosh2017security}
Deepak~K. Tosh, Sachin Shetty, Xueping Liang, Charles~A. Kamhoua, Kevin~A.
  Kwiat, and Laurent Njilla.
\newblock Security implications of blockchain cloud with analysis of block
  withholding attack.
\newblock In {\em Proceedings of the 17th IEEE/ACM International Symposium on
  Cluster, Cloud and Grid Computing}, CCGrid '17, pages 458--467, Piscataway,
  NJ, USA, 2017. IEEE Press.

\bibitem{WU2019Equilibrium}
Di~Wu, Xiangdong Liu, Xiangbin Yan, Rui Peng, and Gang Li.
\newblock Equilibrium analysis of bitcoin block withholding attack: A
  generalized model.
\newblock {\em Reliability Engineering \& System Safety}, 185:318 -- 328, 2019.

\end{thebibliography}
\end{document}